\documentclass[conference]{IEEEtran}

\usepackage{cite}
\usepackage{amsmath,amssymb,amsthm}
\usepackage{microtype}

\usepackage{booktabs}
\usepackage{tabularx}
\usepackage{array}
\usepackage{multirow}

\usepackage{graphicx}
\graphicspath{{figures/}}

\usepackage{balance}
\usepackage{url}
\usepackage{algorithm}
\usepackage{algpseudocode}
\usepackage{tikz}
\usetikzlibrary{arrows.meta, positioning, fit, backgrounds, calc}

\usepackage[dvipsnames]{xcolor}

\definecolor{linkblue}{RGB}{0, 83, 156}    
\definecolor{citegreen}{RGB}{0, 120, 80}   
\definecolor{urlcolor}{RGB}{180, 50, 20}   

\usepackage[
  colorlinks = true,
  linkcolor  = linkblue,
  citecolor  = red,
  urlcolor   = red,
  bookmarks  = true,
  pdftitle   = {Query-Aware Token Budgeting for Efficient
                Late-Interaction Visual Document Retrieval},
  pdfauthor  = {PS Rishi; Rajeev Ranjan Dwivedi; Vinod Kumar Kurmi},
]{hyperref}

\newtheorem{proposition}{Proposition}
\newcommand{\R}{\mathbb{R}}
\newcommand{\MaxSim}{\operatorname{MaxSim}}
\newcommand{\argmax}{\operatorname*{arg\,max}}
\newcommand{\ndcg}{nDCG@5}

\begin{document}

\title{Query-Aware Token Budgeting for Efficient Late-Interaction Visual Document Retrieval}

\author{
\IEEEauthorblockN{
PS Rishi\textsuperscript{*},
Rajeev Ranjan Dwivedi\textsuperscript{*\textdagger},
Vinod Kumar Kurmi
}
\IEEEauthorblockA{
Indian Institute of Science Education and Research Bhopal (IISER Bhopal), India\\
\texttt{rishi.sivakumar10@gmail.com},
\texttt{rajeevias95@gmail.com},
\texttt{vinodkk@iiserb.ac.in}\\
\textsuperscript{*}Equal contribution.
\textsuperscript{\textdagger}Corresponding author.
}
}

\maketitle

\begin{abstract}
 Late-interaction visual document retrievers preserve fine-grained page evidence by storing many token embeddings per page, but the resulting storage and query-time interaction costs make large-scale deployment expensive. Pooling document tokens before indexing offers a natural remedy, yet static pooling must decide which visual evidence to preserve before the query is known. We study an alternative: a heavily compressed hot-path index generates candidates, after which query-aware token budgeting operates on the original token sets of the shortlisted pages. We formulate this stage-two selection as a budgeted MaxSim coverage problem, show that a clipped version is monotone submodular, and compare coverage-only, cluster-guided, token-wise, and marginal-gain policies. On ten ViDoRe tasks with ColModernVBERT, direct static pooling reduces macro normalized discounted cumulative gain at rank five from 0.6309 without compression to 0.4738 at a thirty-two-fold pool factor. Under the same candidate-generation regime and a pool-factor-eight-equivalent reranking budget, token top-k recovers 93.93 percent of the full-token score, while greedy marginal-gain selection recovers 98.39 percent. Held-out and leave-one-dataset-out evaluations yield positive greedy improvements over token top-k on every dataset. The latency analysis reveals two useful operating points: token top-k for interactive retrieval and the naive greedy implementation as a quality upper envelope. Together, these results show that late-interaction visual retrieval benefits from query-aware allocation rather than query-agnostic pooling alone.
\end{abstract}


\section{Introduction}

Document retrieval is increasingly applied to corpora whose meaning is not captured by text alone: forms, scientific papers, financial reports, charts, invoices, slide decks, government documents, and multilingual technical pages. In such corpora, a page is a visual object whose semantics are distributed across words, tables, figures, typography, layout, and spatial grouping. Text-only pipelines based on optical character recognition, layout parsing, captioning, and chunking---including dedicated document-AI models such as LayoutLM \cite{xu2020layoutlm} and OCR-free transformers such as Donut \cite{kim2022donut}---remain useful, but they can discard or distort evidence that is readily visible on the page.

Visual document retrievers address this limitation by embedding pages directly. ColPali combines visual page embeddings with a ColBERT-style MaxSim operator and outperforms conventional OCR-based pipelines on the ViDoRe benchmark \cite{colpali,khattab2020colbert}. Complementary approaches include single-vector Document Screenshot Embedding \cite{ma2024dse} and the visual retrieval-augmented generation pipeline VisRAG \cite{yu2025visrag}, both of which encode page screenshots without parsing. ModernVBERT and ColModernVBERT further show that compact retrieval-oriented vision-language encoders, built on contrastive image--text pretraining with sigmoid losses similar to SigLIP \cite{zhai2023siglip}, can retain much of this retrieval quality at lower model cost \cite{modernvbert}. Although these models simplify visual retrieval, they also expose a systems bottleneck: each page is represented by hundreds of token vectors rather than a single vector.

The cost is structural. For query token embeddings $Q=\{q_i\}_{i=1}^m$ and document token embeddings $D=\{d_j\}_{j=1}^n$, a late-interaction retriever scores
\begin{equation}
\MaxSim(Q,D) = \sum_{i=1}^{m}\max_{1\leq j\leq n} q_i^\top d_j .
\label{eq:maxsim}
\end{equation}
The maximum over document tokens makes the model effective on localized evidence such as table cells, captions, and figure labels, but it also makes storage and scoring expensive. In our ColModernVBERT setting, a 768-pixel page has roughly 350 document tokens. A million-page index therefore requires hundreds of millions of token vectors even before query-time computation is considered. Fig.~\ref{fig:teaser} sketches the alternative we study: retain only the document tokens that fall in the query's feature-context region.

\begin{figure}
    \centering
    \includegraphics[width=1\linewidth]{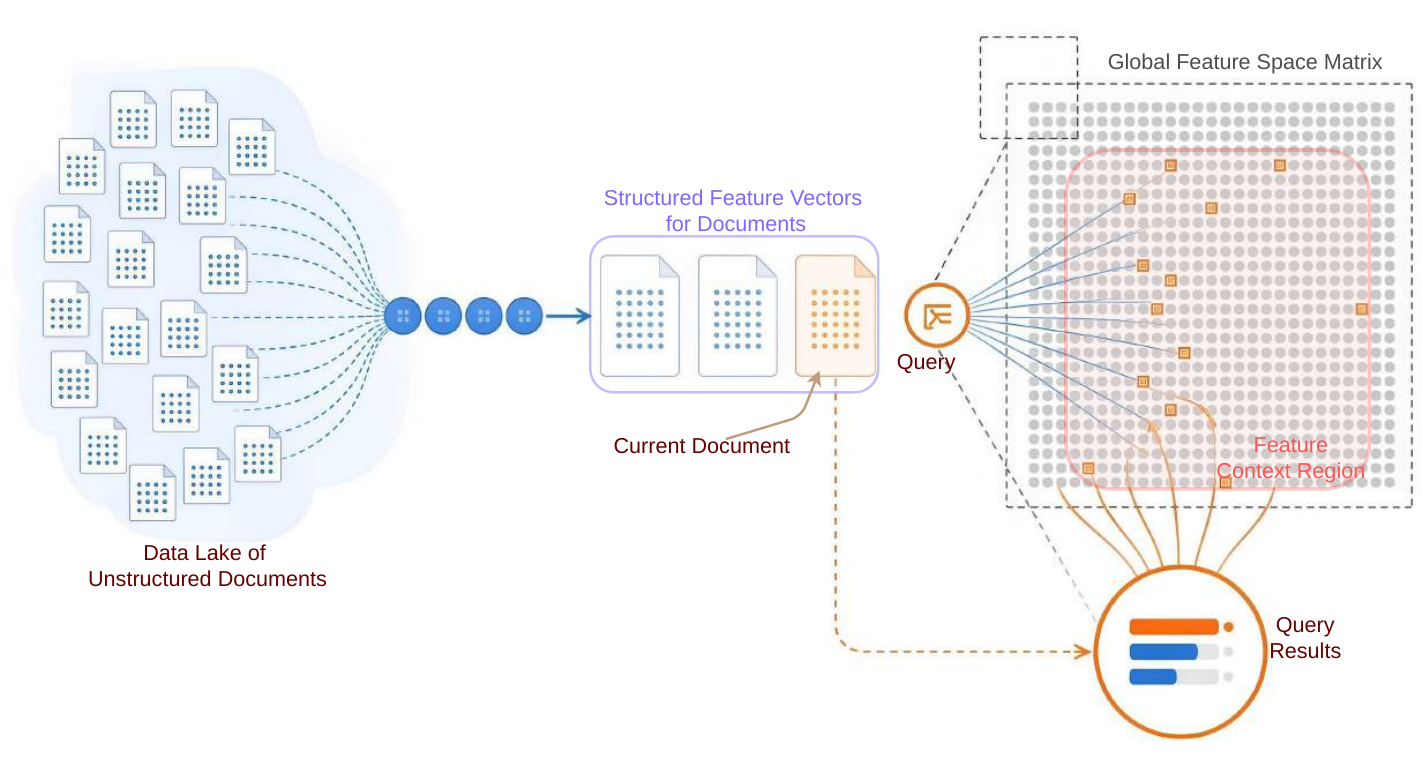}     \vspace{-2em}
    \caption{Query-aware token budgeting. Pages in a document corpus are encoded as structured token vectors. At query time, only the tokens that fall in the query's feature-context region are retained for late-interaction scoring, rather than scoring every token on every page.}
    
    \label{fig:teaser}
    \vspace{-1em}
\end{figure}

A common response is static token pooling, which clusters or aggregates page tokens at indexing time and stores only their representatives \cite{clavie2024pooling}. Related approaches include token pruning that discards uninformative embeddings \cite{lassance2021studytoken,acquavia2023pruning} and product-quantization-based engines that compress token residuals after centroid clustering \cite{santhanam2021colbertv2,santhanam2022plaid,nardini2024emvb}. Static pooling is attractive because it reduces the full index, but a single compressed representation must then serve every future query. This constraint is poorly aligned with MaxSim: the token that matters for a query need not be globally salient on the page, but only the best match for a particular query term. If pooling merges or removes that token, the query-time maximum can change even when the compressed representation remains geometrically plausible.

We therefore study a different allocation of computation that decouples corpus-scale candidate generation from query-specific evidence preservation. Stage one uses a strongly compressed index to retrieve a small candidate set, in the spirit of two-stage neural ranking pipelines that pair an efficient first retriever with a more expensive verifier \cite{nogueira2019passage}. Stage two revisits the original token set only for those candidate pages and selects a fixed budget of tokens for final MaxSim scoring
(Fig.~\ref{fig:pipeline}). The problem becomes: given a query, a candidate page, and a token budget, which document tokens should be retained?

This perspective turns compression into a data mining problem over multi-vector representations: selecting a compact subset of document tokens that collectively covers the query under a strict budget. It also creates a natural connection to submodular coverage. A token that covers an already-represented query aspect has less marginal value than one that covers a missing aspect. Budgeted selection should therefore reward complementarity as well as independent relevance, echoing the use of submodular coverage and maximal marginal relevance in classical information retrieval and summarization \cite{carbonell1998mmr,lin2011submodular,krause2014submodular}.

The greedy submodular algorithm itself is classical. Our contribution is to expose query-conditioned selection within shortlisted visual pages as a MaxSim coverage problem, integrate it with a compressed first-stage index, and quantify the resulting quality--latency frontier under matched token budgets.

We make the following contributions.

\begin{itemize}
    \item We formulate query-aware token budgeting for late-interaction visual document retrieval as a stage-two MaxSim coverage problem under a per-page cardinality budget.
    \item We analyze the clipped token-selection objective as a monotone submodular function and apply classical greedy maximization to obtain a redundancy-aware selector.
    \item We compare random, uniform, cluster-guided Budget-Constrained Reranking (BCR), token top-$k$, and greedy selection under matched budgets on ten ViDoRe tasks using ColModernVBERT.
    \item We report quality, recovery, held-out generalization, leave-one-dataset-out transfer, and latency measurements, showing a clear frontier between token top-$k$ and greedy selection.
\end{itemize}

The main empirical result is that static compression alone is insufficient. Direct pooling from roughly 350 tokens per page to approximately 11 tokens per page reduces macro \ndcg{} from 0.6309 to 0.4738. Under the same compressed first-stage regime, a pool-factor-eight-equivalent second-stage budget recovers most of this loss: token top-$k$ reaches 0.5926, while greedy marginal-gain selection reaches 0.6208, or 98.39 percent of the uncompressed baseline.

\begin{figure*}
    \centering
    \includegraphics[width=1\textwidth]{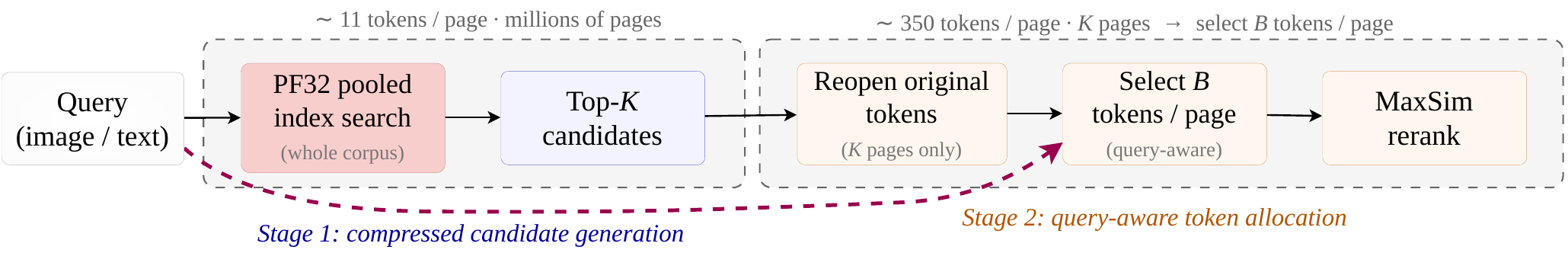}
\caption{Two-stage retrieval pipeline. Stage~1 (blue) searches a PF32-pooled index across the full corpus to retrieve $K$ candidate pages, storing ${\approx}11$ tokens per page and reducing hot-path
storage and interaction cost by $32{\times}$. Stage~2 (orange)
reopens the original ${\approx}350$-token representations for those
$K$ pages and applies query-aware token budgeting to select a budget
of $B$ tokens per page before final MaxSim reranking. The key
asymmetry is that Stage~1 is query-agnostic and corpus-wide, while
Stage~2 is query-specific and restricted to the shortlist.}
  \label{fig:pipeline}
\end{figure*}

\section{Related Work}

\subsection{Text and dense retrieval}
Classical information retrieval relies on sparse lexical signals such as TF-IDF and BM25 \cite{sparckjones1972,robertson1994okapi}. Neural bi-encoders such as Sentence-BERT and Dense Passage Retrieval encode queries and passages into single dense vectors, enabling efficient approximate nearest-neighbor search with libraries such as FAISS \cite{reimers2019sentence,karpukhin2020dense,johnson2019faiss}. Learned sparse retrievers such as SPLADE recover lexical controllability while remaining trainable end-to-end \cite{formal2021splade}. To recover finer-grained relevance signals at a manageable cost, BERT-style cross-encoders are commonly used to rerank a small candidate set produced by a faster first stage \cite{nogueira2019passage}. These models are scalable but can blur fine-grained evidence inside a document. Retrieval-augmented generation further raises the stakes: generation quality depends on whether the retriever surfaces the right evidence in the first place \cite{lewis2020rag}.

\subsection{Late interaction}
Late-interaction retrieval, exemplified by ColBERT, stores a set of contextual token embeddings for each document and scores a query by summing each query token's best document-token match \cite{khattab2020colbert}. ColBERTv2 introduced residual compression and a denoised supervision strategy \cite{santhanam2021colbertv2}, while PLAID further accelerated retrieval with centroid interaction and centroid pruning \cite{santhanam2022plaid}. EMVB augments PLAID with optimized bit-vector prefiltering and SIMD-accelerated centroid interaction \cite{nardini2024emvb}. XTR trains the encoder to surface important document tokens early \cite{lee2023xtr}, and ColBERTer reduces the stored vector count through whole-word aggregation and pruning \cite{hofstatter2022colberter}. Most directly, Col-Bandit adaptively prunes query--document interaction entries at query time without modifying the index \cite{pony2026colbandit}. Our work instead chooses a budgeted subset of original tokens within each shortlisted page before final MaxSim reranking; the two approaches reduce computation at complementary points.

\subsection{Visual document retrieval}
Vision-language models such as CLIP and SigLIP established shared image-text embedding spaces \cite{radford2021clip,zhai2023siglip}, but document retrieval requires finer visual grounding than natural-image retrieval. Earlier document-AI work approached this via OCR-then-language-model pipelines: LayoutLM jointly pretrains text and layout for document image understanding \cite{xu2020layoutlm}, and Donut removes OCR entirely with an end-to-end encoder-decoder \cite{kim2022donut}. ColPali introduced direct visual page retrieval with late interaction and the ViDoRe benchmark, showing that page images can be retrieved without lossy OCR-first pipelines \cite{colpali}. Document Screenshot Embedding pursues a similar paradigm with a single-vector encoder \cite{ma2024dse}, while VisRAG extends visual retrieval into retrieval-augmented generation by directly embedding pages as images \cite{yu2025visrag}. ModernVBERT and ColModernVBERT reduce the model-size burden while preserving late-interaction document retrieval performance \cite{modernvbert}. These models shift the deployment bottleneck from encoder size to token storage and MaxSim scoring, motivating our focus on token budgets.

\subsection{Token reduction and coverage selection}
Token reduction methods include hard pruning, clustering, pooling, and merging. In efficient transformer inference, EViT prunes uninformative vision tokens \cite{liang2022evit}, DynamicViT learns a token sparsification policy end-to-end \cite{rao2021dynamicvit}, and ToMe merges similar tokens without retraining \cite{bolya2023tome}. PixelPrune removes redundant image patches before visual encoding, while QuoTA assigns visual tokens according to a query for long-video comprehension \cite{wang2026pixelprune,luo2026quota}; these methods operate at different stages and on different tasks from late-interaction document reranking. In multi-vector retrieval, token pruning preserves high-value embeddings \cite{lassance2021studytoken,acquavia2023pruning}, while clustering-based pooling merges similar tokens and gives the pool-factor (PF) terminology used here \cite{clavie2024pooling}. PLAID and EMVB instead store centroid identifiers and low-bit residuals \cite{santhanam2022plaid,nardini2024emvb}. These query-agnostic methods cannot know which evidence a later query will need; we use pooling for candidate generation and shift fine-grained allocation to query time.

The distinction is not simply whether tokens are removed, but when the decision is made and what information is available at that point. Index-time reduction can be amortized over all queries and directly shrinks corpus-wide search, but it must preserve a representation suitable for an unknown query distribution. Query-time reduction has access to the current query and can retain otherwise inconspicuous evidence, although applying it across the entire corpus would be prohibitively expensive. Our two-stage design separates these roles: static pooling supplies the corpus-scale filter, and query-aware selection is reserved for the much smaller candidate set. This division preserves the principal efficiency advantage of static compression while allowing the final representation to depend on the information need.

Submodular optimization provides a principled language for coverage under budgets. For monotone submodular maximization with a cardinality constraint, greedy selection obtains the classical $(1-1/e)$ approximation guarantee \cite{nemhauser1978analysis}; broader treatments and applications are surveyed in \cite{krause2014submodular}. Lin and Bilmes used submodular coverage and diversity terms to construct principled extractive summaries \cite{lin2011submodular}, while the closely related Maximal Marginal Relevance criterion of Carbonell and Goldstein combines query relevance with information novelty for diversity-aware reranking \cite{carbonell1998mmr}. In dense retrieval, DISCO casts multi-vector retrieval as collective query coverage and uses submodular structure to guide efficient subset selection at the corpus level \cite{disco}. We bring this coverage perspective inside a single visual document page, where the selected elements are document tokens rather than corpus items. This shift from corpus-level selection to within-page allocation motivates the formulation developed next.

\section{Query-Aware Token Budgeting}

\subsection{Problem setting}
For each candidate page $d$, let $D=\{d_1,\ldots,d_n\}\subset\R^r$ be the original document-token set and let $Q=\{q_1,\ldots,q_m\}\subset\R^r$ be the query-token set. Static pooling \cite{clavie2024pooling} stores a query-agnostic representation $\widetilde{D}_p$ with about $n/p$ tokens, where $p$ is the pool factor. In our two-stage pipeline, stage one searches a PF32 index to obtain a candidate set $C_K(Q)$. Stage two reopens each $d\in C_K(Q)$ and selects a subset $S\subseteq D$ with $|S|\leq B$, where $B$ corresponds to a target pool-factor-equivalent budget such as PF8, as illustrated in Fig.~\ref{fig:pipeline}.

\emph{Storage architecture.} This design maintains two storage tiers. The PF32 pooled index
(${\approx}\,11$ tokens per page) resides on fast storage and handles first-stage retrieval across
the full corpus; relative to PF1, it uses $32{\times}$ fewer token vectors and lowers
first-stage MaxSim interaction FLOPs by the same factor. The full token set (${\approx}\,350$
tokens per page) is retained on slower or cold storage and is loaded only for the $K$ shortlisted
candidates per query. For $N$ pages, a PF1-only store contains about $350N$ vectors, whereas
separate full-token and PF32 tiers contain about $(350+11)N=361N$ vectors, or $1.03{\times}$
the PF1 vector count before precision and metadata are considered. Thus, the $32{\times}$ claim
applies to the hot-path Stage-1 index, not total system storage. Stage 2 reads up to $350K$
full-token vectors per query (17,500 at $K=50$ and 35,000 at $K=100$) before selection. This
two-tier cost is the price of recovering evidence that static pooling permanently discards. The ideal stage-two subset maximizes the retained MaxSim score:
\begin{equation}
S^*(Q,D;B)=\argmax_{S\subseteq D,\, |S|\leq B}\sum_{i=1}^{m}\max_{d\in S}q_i^\top d .
\label{eq:budgeted}
\end{equation}
This objective is query-aware because the selected subset changes with $Q$. It is also coverage-oriented: a good subset must support all informative query tokens, not merely contain individually high-scoring page tokens.

The formulation separates the token budget from the number of tokens originally produced by the encoder. For a fixed $B$, every shortlisted page receives the same maximum interaction budget at final scoring, making policies comparable by the evidence they preserve rather than by the number of vectors they retain. The selected subset is also allowed to differ across queries for the same page. A token representing a table header, axis label, or localized phrase may be dispensable for one query yet decisive for another; committing to a single reduced representation at indexing time cannot express this conditional relevance. Query-aware budgeting instead treats the original page representation as a reservoir of evidence and allocates the limited reranking capacity to the parts most useful for the current query.

The two stages consequently address different sources of retrieval error. Candidate generation determines which pages remain eligible for reranking; once a relevant page is absent from $C_K(Q)$, no stage-two policy can recover it. Conditional on the same candidate set, however, the selection problem isolates how effectively a policy preserves the evidence needed to order those candidates. Using a common PF32 first stage and matched reranking budgets therefore separates improvements in token allocation from improvements caused by searching more pages or scoring more final-stage vectors. This separation is central to the empirical comparisons below.

\subsection{Selection policies}
We compare five policies that differ in query awareness and redundancy awareness. \emph{Random reopening} spends the budget without using the query or page structure. It is a diagnostic lower bound. \emph{Uniform reopening} spreads the budget broadly across the compressed representation or token groups. It is query-agnostic but coverage-preserving, and therefore a stronger baseline than random selection. \emph{Budget-Constrained Reranking} (BCR) is a diagnostic intermediate between uniform and token-wise selection: a coverage ratio $\rho\in[0,1]$ reserves approximately $\rho B$ tokens for broad, uniform coverage, and the remaining budget reopens groups whose members have stronger query affinity, with the final set capped at $B$ tokens. BCR isolates whether coarse group-level query guidance is sufficient; it is not intended as a state-of-the-art literature baseline. \emph{Token top-$k$} scores each original document token independently by its best query-token similarity,
\begin{equation}
\sigma(d\mid Q)=\max_{1\leq i\leq m} q_i^\top d,
\label{eq:topk}
\end{equation}
then keeps the $B$ highest-scoring tokens. This is fast and strongly query-aware, but it does not penalize redundancy: several selected tokens may cover the same query aspect, much as classical relevance-only rerankers select redundant documents before MMR-style diversification \cite{carbonell1998mmr}.

\emph{Greedy marginal-gain selection} directly optimizes coverage. Starting with $S_0=\emptyset$, at step $t$ it chooses
\begin{equation}
d^*_{t+1}=\argmax_{d\in D\setminus S_t}\left[F_Q(S_t\cup\{d\})-F_Q(S_t)\right],
\label{eq:greedy}
\end{equation}
where $F_Q(S)$ is the MaxSim-style coverage utility defined below. This method is redundancy-aware because a token receives credit only for the query-token coverage it adds beyond the current set.

Token top-$k$ and greedy therefore answer related but different questions. Token top-$k$ asks which document tokens look most relevant in isolation. Greedy asks which token is most useful after accounting for what the current subset already explains. Consider a query with several semantic or visual aspects. Independent scoring may devote much of the budget to near-duplicate tokens that all match the same dominant aspect, leaving weaker but complementary evidence uncovered. Marginal-gain selection suppresses this repetition: once an aspect is well covered, another similar token contributes little, allowing the budget to move to a different query aspect. This interpretation is especially natural for document pages, where repeated patches, neighboring text tokens, and visually similar regions can create many individually strong but redundant matches.

Algorithms~\ref{alg:greedy} and~\ref{alg:topk} give complete pseudocode for the two primary policies.
Table~\ref{tab:complexity} summarizes per-page computational costs; the observed stage-two latency
ratio of ${\approx}\,34{\times}$ for naive greedy over Token top-$k$ (Table~\ref{tab:design}) is consistent
with the $O(B)$ multiplicative factor at $B\approx 45$.

\begin{algorithm}[t]
\caption{Greedy Marginal-Gain Token Selection}
\label{alg:greedy}
\begin{algorithmic}[1]
\Require Query tokens $Q=\{q_i\}_{i=1}^{m}$; document tokens $D=\{d_j\}_{j=1}^{n}$; budget $B$
\Ensure $S\subseteq D$ with $|S|\leq B$
\State $\phi_i(d)\leftarrow\max\{0,\,q_i^\top d\}$ for all $i,d$ \Comment{clipped similarity}
\State $S\leftarrow\emptyset$;\quad $c_i\leftarrow 0$ for all $i$ \Comment{per-query-token coverage}
\For{$t=1,\ldots,B$}
  \For{each $d\in D\setminus S$}
    \State $\Delta(d)\leftarrow\sum_{i=1}^{m}\max\!\bigl\{0,\;\phi_i(d)-c_i\bigr\}$
  \EndFor
  \State $d^*\leftarrow\argmax_{d\in D\setminus S}\;\Delta(d)$ \Comment{ties follow token order}
  \If{$\Delta(d^*)=0$}
    \State \textbf{break} \Comment{clipped coverage has saturated}
  \EndIf
  \State $S\leftarrow S\cup\{d^*\}$;\quad $c_i\leftarrow\max\{c_i,\,\phi_i(d^*)\}$ for all $i$
\EndFor
\Return $S$
\end{algorithmic}
\end{algorithm}

\begin{algorithm}[t]
\caption{Token Top-$k$ Selection}
\label{alg:topk}
\begin{algorithmic}[1]
\Require Query tokens $Q=\{q_i\}_{i=1}^{m}$; document tokens $D=\{d_j\}_{j=1}^{n}$; budget $B$
\Ensure $S\subseteq D$ with $|S|=B$
\For{each $d\in D$}
  \State $\sigma(d)\leftarrow\max_{1\leq i\leq m}\;q_i^\top d$
\EndFor
\Return top-$B$ elements of $D$ ranked by $\sigma(\cdot)$
\end{algorithmic}
\end{algorithm}

\begin{table}[t]
\centering 
\caption{Per-page time complexity at stage two. $n$: document tokens; $m$: query tokens; $B$: budget; $k$: clusters ($k\approx n/32$ under PF32 indexing). Naive greedy is a factor of $B$ more expensive than token top-$k$; lazy greedy has the same worst-case bound but can be faster in practice.}
\label{tab:complexity}
\begin{tabular}{lll}
\toprule
Method & Time & Space \\
\midrule
Random             & $O(B)$                                        & $O(1)$     \\
Uniform            & $O(n)$                                        & $O(1)$     \\
Token top-$k$      & $O(nm)$                                       & $O(n)$     \\
BCR                & $O(km+n)$                                     & $O(k)$     \\
Greedy (naive)     & $O(Bnm)$                                      & $O(n{+}m)$ \\
Lazy greedy        & $O(Bnm)$ worst case                           & $O(n{+}m)$ \\
\bottomrule
\end{tabular}
\end{table}

\begin{table}[t]
\centering 
\caption{Selection policies compared in the reranking stage.}
\label{tab:methods}
\begin{tabular}{lccc}
\toprule
Policy & Query-aware & Redundancy-aware & Unit \\
\midrule
Random & No & No & token/cluster \\
Uniform & No & Implicit & token/cluster \\
BCR & Yes & Coarse & cluster \\
Token top-$k$ & Yes & No & token \\
Greedy & Yes & Yes & token \\
\bottomrule
\end{tabular}
\end{table}

\subsection{Submodular structure}
The raw inner product in (\ref{eq:maxsim}) may be negative. For the selection objective, we use the standard nonnegative coverage version obtained by including a dummy zero token, equivalently $\phi_i(d)=\max\{0,q_i^\top d\}$. Define
\begin{equation}
F_Q(S)=\sum_{i=1}^{m}\max_{d\in S}\phi_i(d),\qquad F_Q(\emptyset)=0 .
\label{eq:coverage}
\end{equation}
The final reranking score is computed with the selected tokens using the ordinary MaxSim operator.

\begin{proposition}
\label{prop:1}
For a fixed query $Q$, the set function $F_Q(S)$ in (\ref{eq:coverage}) is monotone submodular over document-token subsets. Therefore, greedy selection under a cardinality budget $B$ achieves the classical $(1-1/e)$ approximation guarantee for maximizing $F_Q$.
\end{proposition}

\begin{proof}
For a single query token $q_i$, define $f_i(S)=\max_{d\in S}\phi_i(d)$ with $f_i(\emptyset)=0$. If $A\subseteq B$, then $f_i(A)\leq f_i(B)$, so $f_i$ is monotone. For any token $x\notin B$, the marginal gain is
\begin{equation}
f_i(A\cup\{x\})-f_i(A)=\max\{0,\phi_i(x)-f_i(A)\}.
\end{equation}
Since $f_i(A)\leq f_i(B)$, we have $\phi_i(x)-f_i(A)\geq\phi_i(x)-f_i(B)$, and therefore $\max\{0,\phi_i(x)-f_i(A)\}\geq\max\{0,\phi_i(x)-f_i(B)\}$, i.e., $f_i(A\cup\{x\})-f_i(A)\geq f_i(B\cup\{x\})-f_i(B)$. This is the submodularity (diminishing-returns) condition, so $f_i$ is submodular. A nonnegative sum of monotone submodular functions is monotone submodular, so $F_Q=\sum_i f_i$ is monotone submodular. The greedy guarantee follows from Nemhauser et al. \cite{nemhauser1978analysis}; see also the survey of submodular function maximization in \cite{krause2014submodular}.
\end{proof}

The proposition clarifies the difference between token top-$k$ and greedy selection. Token top-$k$ estimates token salience independently. Greedy estimates marginal coverage, which is the relevant quantity under MaxSim when tokens can be redundant.

\emph{Remark (selection vs.\ scoring objective).} Algorithm~\ref{alg:greedy} maximizes the
clipped objective $F_Q$ defined in~\eqref{eq:coverage}, for which the submodular guarantee
holds, and stops once no remaining token has positive marginal gain. Final reranking in
Tables~\ref{tab:design}--\ref{tab:latency} uses the standard MaxSim operator~\eqref{eq:maxsim},
which permits negative inner products. Clipping prevents negative similarities from creating
positive selection utility, but it does not make the two objectives identical: a selected token
may still have negative inner products with individual query tokens. Accordingly, the
$(1-1/e)$ guarantee applies to $F_Q$, not directly to ordinary MaxSim or retrieval nDCG.

\section{Experimental Setup}

\subsection{Benchmark and model}
We evaluate on ten ViDoRe tasks \cite{colpali}: DocVQA \cite{mathew2021docvqa}, InfoVQA \cite{mathew2022infographicvqa}, TAT-DQA \cite{zhu2022tatdqa}, ArXivQA, TabFQuAD, Energy, Government, Healthcare, Artificial Intelligence, and Shift Project. Together, these tasks span scanned industrial documents, infographics, financial tables, scientific figures, web tables, government and medical pages, energy reports, AI-related scientific pages, and French environmental reports. Their query counts range from 100 for the practical tasks to 1600 for TAT-DQA.

All experiments use ColModernVBERT \cite{modernvbert} as the base visual retriever. The model preserves ColPali's multi-vector late-interaction structure while using a compact 250M-parameter encoder. In the comparison reported with ModernVBERT, ColModernVBERT obtains a ViDoRe average of 68.6 with a CPU query-encoding latency of 0.032 seconds, compared with 69.2 and 0.222 seconds for ColPali \cite{modernvbert}. This balance shifts the remaining efficiency challenge toward token storage and interaction.

\begin{table}[t]
\centering
\caption{ViDoRe tasks used in the evaluation~\cite{colpali}.
         Query and document counts follow the ColPali protocol.
         The ten tasks span scanned industrial pages, infographics, financial
         and web tables, scientific figures, and multilingual reports.
         DocVQA~\cite{mathew2021docvqa}, InfoVQA~\cite{mathew2022infographicvqa},
         and TAT-DQA~\cite{zhu2022tatdqa} have dedicated dataset papers.}
\label{tab:datasets}
\begin{tabular}{lrrl}
\toprule
Dataset & Queries & Docs & Characteristic \\
\midrule
DocVQA & 500 & 500 & scanned industrial docs \\
InfoVQA & 500 & 500 & web infographics \\
TAT-DQA & 1600 & 1600 & financial tables \\
ArXivQA & 500 & 500 & scientific figures \\
TabFQuAD & 210 & 210 & French web tables \\
Energy & 100 & 1000 & energy reports \\
Government & 100 & 1000 & administrative pages \\
Healthcare & 100 & 1000 & medical pages \\
AI & 100 & 1000 & AI scientific pages \\
Shift Project & 100 & 1000 & French environmental reports \\
\bottomrule
\end{tabular}
\end{table}

\subsection{Budgets and metrics}
\label{sec:budgets}
The uncompressed PF1 representation has about 350 page tokens. Static compression following \cite{clavie2024pooling} stores pooled representations at PF4, PF8, PF16, and PF32, corresponding roughly to 88, 45, 22, and 11 tokens per page. The main two-stage experiments use PF32 for candidate generation and a PF8-equivalent reranking budget unless otherwise stated. The full ten-dataset design-space comparison uses $K=50$ candidates; held-out and latency sensitivity analyses additionally consider $K=20$ and $K=100$.

Expressing the reranking budget through PF-equivalent token counts provides a direct comparison with static pooling. For example, a PF8-equivalent budget allows a query-aware method to score roughly the same number of final-stage tokens as a PF8 pooled representation, while differing only in how those tokens are chosen. The comparison therefore asks whether query conditioning and redundancy control make better use of a fixed interaction budget, rather than giving the proposed policies additional final-stage capacity.

The main quality metric is macro-averaged \ndcg{} over datasets. We also report recovery relative to the uncompressed PF1 score:
\begin{equation}
\mathrm{Recovery}(M)=100\cdot\frac{\ndcg(M)}{\ndcg(\mathrm{PF1})}.
\label{eq:recovery}
\end{equation}
Latency is reported as median stage-two time and total median query time on the measured latency subset. Stage-two time covers token selection and final reranking; total time adds the measured first stage. Cold-storage reads, deserialization, and host-to-device transfer were not profiled separately, so these values characterize the measured retrieval pipeline rather than an end-to-end storage-tier deployment. In all empirical results, ``Greedy'' denotes the naive implementation of Algorithm~\ref{alg:greedy}; lazy evaluation and other accelerations are not measured.

\begin{figure*}[t]
\centering
\includegraphics[width=0.85\textwidth]{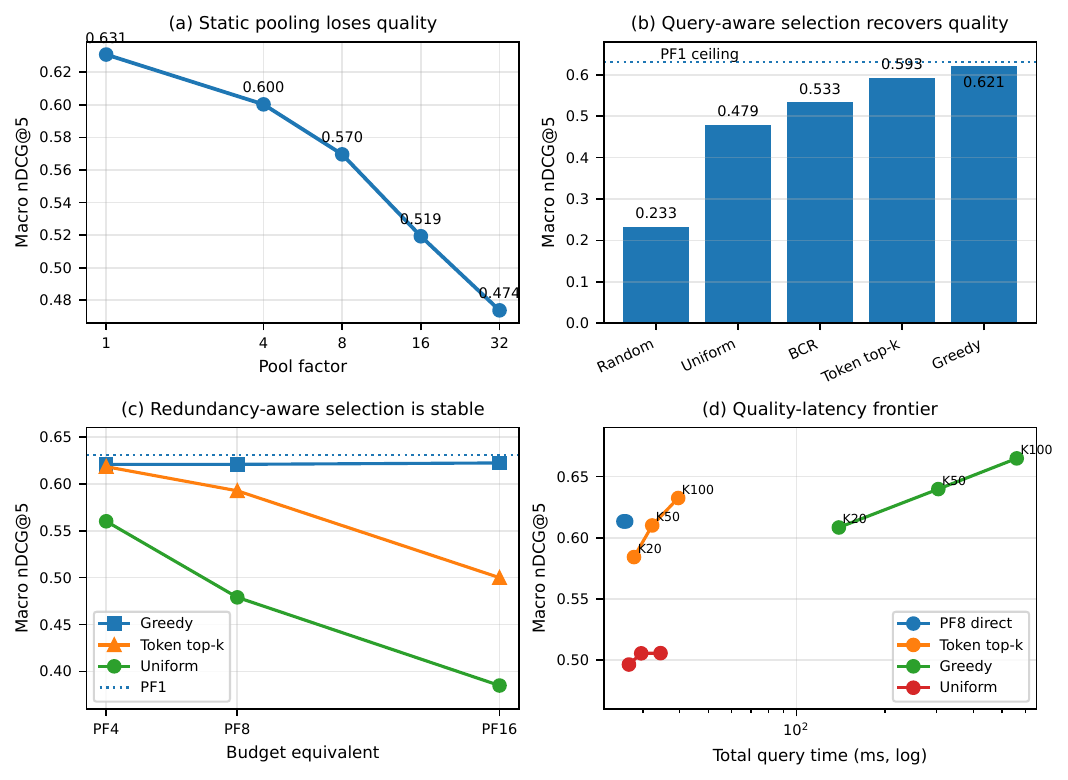}
\caption{Main empirical trends. Static pooling loses quality as compression increases. Query-aware selection recovers most of the PF1 score under a PF8-equivalent budget. Token top-$k$ is the low-latency operating point; naive greedy is the quality envelope.}
\label{fig:overview}
\end{figure*}

\section{Results}

\subsection{Static pooling has a hard ceiling}
Table~\ref{tab:compression} reports the direct-compression baseline. Performance declines monotonically as the pool factor increases: PF8 retains 90.28 percent of PF1, whereas PF32 retains only 75.11 percent. This exposes the central limitation of query-agnostic compression: a representation that is efficient for candidate generation can become too lossy for final ranking. The pattern is consistent with the original findings on clustering-based token pooling \cite{clavie2024pooling}, although our ratios are measured with a visual rather than textual late-interaction backbone.

\begin{table}[t]
\centering
\caption{Direct static pooling on ten ViDoRe tasks (macro nDCG@5).}
\label{tab:compression}
\begin{tabular}{lrrr}
\toprule
Pool factor & Tokens/page & Macro \ndcg{} & Drop vs. PF1 \\
\midrule
PF1  & 350 & 0.6309 & 0.00\% \\
PF4  & 88  & 0.6003 & 4.85\% \\
PF8  & 45  & 0.5696 & 9.72\% \\
PF16 & 22  & 0.5193 & 17.68\% \\
PF32 & 11  & 0.4738 & 24.89\% \\
\bottomrule
\end{tabular}
\end{table}

The left panel of Fig.~\ref{fig:overview} visualizes this trend. The sharp loss at PF32 motivates the two-stage design: use the compressed representation to retrieve candidates efficiently, then return to finer evidence for final ranking, following the broader logic of neural reranking pipelines \cite{nogueira2019passage}.

\subsection{Query-aware token budgeting recovers quality}
Table~\ref{tab:design} compares the selection policies under a PF8-equivalent stage-two budget. Their ordering follows the coverage interpretation: random $<$ uniform $<$ BCR $<$ token top-$k$ $<$ greedy. Random selection performs poorly because an unstructured budget rarely preserves the relevant evidence. Uniform reopening is substantially stronger, showing that broad coverage helps even without query information, and BCR improves further by adding coarse query guidance. Token top-$k$ then benefits from token-level query affinity and provides the practical interactive choice. Naive greedy achieves the highest score by accounting for redundancy through marginal gain, but its ${\approx}34{\times}$ higher stage-two latency makes it a quality reference rather than the default deployment policy.

\begin{table}[t]
\centering
\caption{Design-space comparison at PF8-equivalent budget and $K=50$. Recovery is relative to PF1.}
\label{tab:design}
\scalebox{1.0}{\begin{tabular}{lrrr}
\toprule
Method & Macro \ndcg{} & Recovery & Stage-2 ms \\
\midrule
Random & 0.2332 & 36.97\% & 12.428 \\
Uniform & 0.4790 & 75.93\% & 10.550 \\
BCR ($\rho=0.75$) & 0.5333 & 84.53\% & 30.736 \\
Token top-$k$ & 0.5926 & 93.93\% & 15.832 \\
Naive greedy & 0.6208 & 98.39\% & 538.480 \\
\bottomrule
\end{tabular}}
\end{table}

The improvement from token top-$k$ to greedy is especially informative because both methods use query information at token granularity. Greedy differs by discounting tokens whose query aspects are already covered. The increase from 0.5926 to 0.6208 therefore isolates the value of redundancy control rather than access to additional query information. This behavior parallels the advantage of diversity-aware criteria such as MMR over relevance-only rankers \cite{carbonell1998mmr} and the role of submodular diversity in extractive summarization \cite{lin2011submodular}.

\subsection{Greedy is robust across budgets}
Table~\ref{tab:budget} evaluates PF4-, PF8-, and PF16-equivalent reranking budgets. Greedy remains nearly constant across this range, with scores close to 0.621. Token top-$k$ is competitive at PF4 but declines at PF8 and PF16, while uniform reopening declines more sharply. Independent token affinity therefore works well when the budget is generous enough to absorb redundancy; as the budget tightens, selection by marginal coverage becomes increasingly valuable.

\begin{table}[t]
\centering 
\caption{Selection quality across reranking budgets at $K=50$.}
\label{tab:budget}
\begin{tabular}{lrrrr}
\toprule
Budget & PF1 & Greedy & Top-$k$ & Uniform \\
\midrule
PF4  & 0.6309 & 0.6208 & 0.6183 & 0.5601 \\
PF8  & 0.6309 & 0.6208 & 0.5926 & 0.4790 \\
PF16 & 0.6309 & 0.6223 & 0.5000 & 0.3850 \\
\bottomrule
\end{tabular}
\end{table}

The budget sweep also clarifies that the policies fail in different ways. Uniform reopening loses quality because it cannot direct capacity toward query-relevant regions. Token top-$k$ corrects this mismatch, but a tighter budget makes repeated high-affinity tokens increasingly costly. Greedy remains stable because its objective values a token only through the coverage it adds to the current subset. Thus, the benefit of marginal-gain selection is not merely a better ranking of individual tokens; it is a more effective distribution of limited capacity across the query aspects represented on the page.

\subsection{Cluster-level guidance helps but is limited}
BCR interpolates between uniform coverage and query-guided cluster reopening. Fig.~\ref{fig:bcr} shows that pure query-guided cluster reopening is brittle, while near-pure coverage underuses the query. The best setting is $\rho=0.75$, where most of the budget protects broad coverage and the remaining budget adapts to the query. This confirms that cluster-level query information is useful, but still less flexible than token-level selection.

\begin{figure}[h]
\centering
\includegraphics[width=0.9\columnwidth]{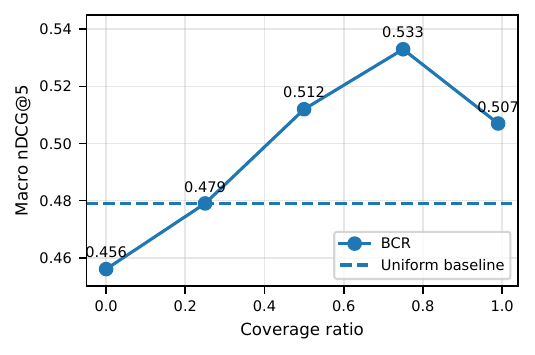}
\caption{BCR coverage-ratio ablation at PF8-equivalent budget and $K=50$. Best at $\rho=0.75$.}
\label{fig:bcr}
\end{figure}

\subsection{Held-out and transfer results}
To assess whether the greedy advantage results from tuning on the full benchmark, we tune on three validation datasets and evaluate on seven held-out datasets. Greedy reaches 0.6460 macro \ndcg{}, compared with 0.6195 for token top-$k$. Relative to the held-out PF1 ceiling of 0.649, these values correspond to 99.54 percent recovery for greedy and 95.46 percent for token top-$k$. Greedy wins on all seven held-out datasets. A Wilcoxon signed-rank test gives $p=0.0156$, a paired $t$ test gives $p=0.0358$, and a bootstrap 95 percent confidence interval for the mean delta is $[0.0096,0.0453]$.

\begin{table}[h]
\centering
\caption{Held-out seven-dataset comparison after tuning on a disjoint validation split.}
\label{tab:heldout}
\begin{tabular}{lrrr}
\toprule
Method & Macro \ndcg{} & Recovery & Wins \\
\midrule
Greedy & 0.6460 & 99.54\% & 7/7 \\
Token top-$k$ & 0.6195 & 95.46\% & 0/7 \\
\bottomrule
\end{tabular}
\end{table}

Fig.~\ref{fig:generalization} shows the per-dataset differences. The largest held-out gains occur on ArXivQA and Shift Project, where relevant evidence can be distributed across visually distinct regions. Leave-one-dataset-out transfer over all ten datasets yields the same qualitative result: greedy wins in every held-out condition, with a mean delta of 0.0282, a minimum of 0.0026, and a maximum of 0.0603.

The two evaluations provide complementary evidence. The fixed held-out split tests whether the observed advantage persists after choices are made away from the evaluation datasets, while leave-one-dataset-out transfer asks whether the direction of the improvement depends on any single benchmark domain. The gains vary in magnitude, as expected from the visual and semantic diversity of ViDoRe, but their consistently positive direction supports the underlying mechanism: redundancy-aware coverage remains useful across pages whose relevant evidence is organized in markedly different ways.

\begin{figure*}[t]
\centering
\includegraphics[width=0.8\textwidth]{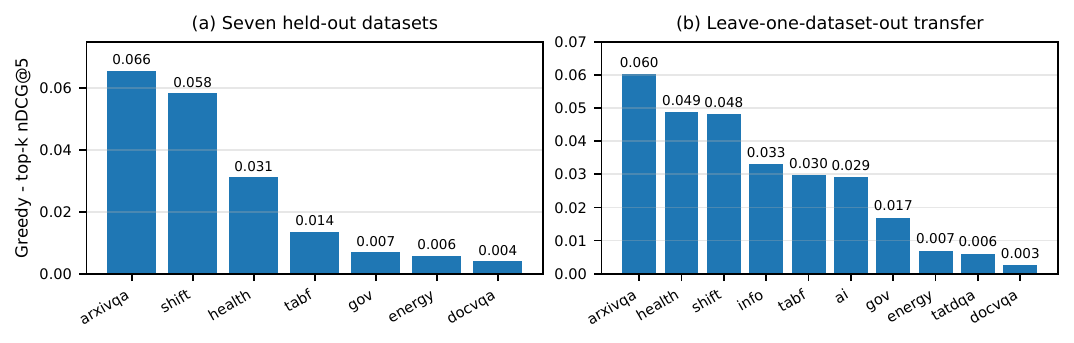}
\caption{Generalization of greedy over token top-$k$. (a)~All seven held-out deltas are positive. (b)~Leave-one-dataset-out transfer is positive on all ten datasets.}
\label{fig:generalization}
\end{figure*}

\begin{table}[h]
\centering
\caption{Per-dataset greedy versus token top-$k$ deltas. Dashes mark validation-split datasets.}
\label{tab:perdataset}
\small
\scalebox{0.80}{\begin{tabular}{lrrr rrr}
\toprule
Dataset & \multicolumn{3}{c}{Held-out protocol} &
          \multicolumn{3}{c}{Leave-one-out transfer} \\
\cmidrule(lr){2-4}\cmidrule(lr){5-7}
 & Greedy & Top-$k$ & $\Delta$ & Greedy & Top-$k$ & $\Delta$ \\
\midrule
ArXivQA~\cite{colpali}         & 0.7044 & 0.6388 & 0.0656 & 0.6946 & 0.6343 & 0.0603 \\
Healthcare~\cite{colpali}      & 0.7939 & 0.7628 & 0.0311 & 0.8021 & 0.7533 & 0.0488 \\
Shift Project~\cite{colpali}   & 0.5574 & 0.4991 & 0.0584 & 0.5372 & 0.4891 & 0.0482 \\
InfoVQA~\cite{mathew2022infographicvqa} & -- & -- & -- & 0.6761 & 0.6429 & 0.0331 \\
TabFQuAD~\cite{colpali}        & 0.4202 & 0.4066 & 0.0136 & 0.4149 & 0.3852 & 0.0297 \\
AI~\cite{colpali}              & --     & --     & --     & 0.8427 & 0.8136 & 0.0291 \\
Government~\cite{colpali}      & 0.7790 & 0.7720 & 0.0069 & 0.7993 & 0.7824 & 0.0169 \\
Energy~\cite{colpali}          & 0.8133 & 0.8076 & 0.0057 & 0.8016 & 0.7946 & 0.0070 \\
TAT-DQA~\cite{zhu2022tatdqa}   & --     & --     & --     & 0.1967 & 0.1906 & 0.0061 \\
DocVQA~\cite{mathew2021docvqa} & 0.4540 & 0.4500 & 0.0040 & 0.4424 & 0.4398 & 0.0026 \\
\bottomrule
\end{tabular}}
\end{table}

\subsection{Latency-quality trade-off}
Table~\ref{tab:latency} reports the measured latency profile on the three-dataset latency subset at $K=100$. Naive greedy gives the best quality but recomputes marginal gains sequentially, taking 538.480 ms in Stage 2 and 562.443 ms in total. Token top-$k$ takes 15.832 ms in Stage 2 and 39.776 ms in total, a roughly $34{\times}$ reduction in selection-stage latency while improving over PF8-direct on this subset. Thus, token top-$k$ is the practical interactive default; naive greedy is a latency-tolerant quality ceiling and a possible teacher for future approximate selectors.

\begin{table}[t]
\centering
\caption{Latency-quality trade-off on the three-dataset latency subset at $K=100$. Total time includes both stages.}
\label{tab:latency}
\scalebox{1.0}{\begin{tabular}{lrrr}
\toprule
Method & Macro \ndcg{} & Stage-2 ms & Total ms \\
\midrule
PF8-direct & 0.6133 & -- & 25.678 \\
PF32 + Uniform & 0.5056 & 10.550 & 34.490 \\
PF32 + Random & 0.2270 & 12.428 & 36.366 \\
PF32 + Token top-$k$ & 0.6324 & 15.832 & 39.776 \\
PF32 + BCR & 0.5091 & 30.736 & 54.671 \\
PF32 + Naive greedy & 0.6648 & 538.480 & 562.443 \\
\bottomrule
\end{tabular}}
\end{table}

\begin{table*}[t]
\centering
\caption{Latency sensitivity across candidate-pool sizes on the three-dataset latency subset. Each cell reports macro nDCG@5 / total median milliseconds. Token top-$k$ remains on the low-latency frontier as $K$ grows.}
\label{tab:ksweep}
\scalebox{1.0}{
\begin{tabular}{lccc}
\toprule
Method & $K=20$ & $K=50$ & $K=100$ \\
\midrule
PF8-direct & 0.6133 / 26.301 & 0.6133 / 25.981 & 0.6133 / 25.732 \\
PF32 + Uniform & 0.4963 / 26.873 & 0.5055 / 29.584 & 0.5056 / 34.430 \\
PF32 + Random & 0.2646 / 27.277 & 0.2414 / 30.593 & 0.2270 / 36.301 \\
PF32 + Token top-$k$ & 0.5842 / 27.958 & 0.6100 / 32.251 & 0.6324 / 39.542 \\
PF32 + BCR & 0.4878 / 31.101 & 0.4953 / 39.767 & 0.5091 / 54.384 \\
PF32 + Greedy & 0.6083 / 139.259 & 0.6397 / 303.635 & 0.6648 / 561.438 \\
\bottomrule
\end{tabular}}
\end{table*}

The latency sweep in Fig.~\ref{fig:overview} and Table~\ref{tab:ksweep} further shows that token top-$k$ remains on the low-latency frontier as $K$ varies, while naive greedy moves to much higher latency but also higher quality. The result does not suggest replacing interactive retrieval with naive greedy. Instead, it shows that query-aware token budgeting offers a useful family of operating points, complementing engine-level optimizations such as PLAID and EMVB \cite{santhanam2022plaid,nardini2024emvb}.

\section{Discussion}

Taken together, the experiments support three lessons.

First, static pooling is useful for candidate generation but insufficient as the final representation under aggressive budgets. PF32 indexing is attractive because it stores roughly 11 tokens per page, yet its standalone score is only 0.4738. Revisiting the shortlisted pages with a modest query-aware budget recovers much of this lost quality. The result complements static-compression methods \cite{clavie2024pooling,acquavia2023pruning,lassance2021studytoken,hofstatter2022colberter}: compression strengthens the first stage, while query-aware selection addresses the resulting final-ranking ceiling.

Second, the key resource is not token count alone, but coverage of the query. Uniform reopening reveals the value of broad coverage; BCR adds coarse query guidance; token top-$k$ introduces token-level affinity; and greedy adds explicit redundancy control. Their monotonic ordering in the main experiment follows the coverage view and mirrors a long-standing insight from summarization and diversification: relevance is most effective when balanced with novelty \cite{carbonell1998mmr,lin2011submodular}.

Third, the preferred method depends on the latency requirement. Naive greedy is not the interactive deployment policy; instead, it defines a principled quality envelope and exposes the cost of ignoring redundancy. Token top-$k$ is less theoretically complete but substantially more practical, while still recovering 93.93 percent of PF1 under the main ten-dataset protocol.

Together, these observations suggest that compression and selection should be treated as complementary decisions rather than interchangeable operations. Compression determines the representation that can be searched economically across the corpus, whereas selection determines how the available fine-grained evidence is used once the candidate space has narrowed. Conflating the two forces a single representation to satisfy competing requirements: it must be compact enough for large-scale search yet detailed enough for every possible query. The two-stage formulation avoids this tension by allowing each representation to serve a distinct role, with query-aware allocation connecting efficient candidate generation to evidence-sensitive final ranking.

\subsection{Theoretical bound versus empirical recovery}

Proposition~\ref{prop:1} guarantees $F_Q(S_{\text{greedy}})\geq(1-1/e)\cdot F_Q(S^*)$ for the optimal
size-$B$ subset $S^*$, a worst-case bound of approximately 63\%. The reported 98.39\%
\textit{nDCG@5} recovery against the PF1 full-token baseline is not an instance-specific
approximation ratio for this objective.

These quantities are not directly comparable: the guarantee concerns the clipped selection
objective relative to the optimal size-$B$ subset $S^*$, whereas 98.39\% is final nDCG recovery
relative to PF1. Within $F_Q$, the gap between a size-$B$ subset and the full token set depends
on the marginal coverage of the $n-B$ excluded tokens. In our evaluated setting, once the
dominant query aspects are covered, the remaining tokens add little useful evidence, so the
final ranking can remain close to PF1. This is an empirical explanation, not an nDCG guarantee.

Moreover, the adversarial instances that force the $(1-1/e)$ tightness require that every greedy
selection maximally duplicates existing coverage. Visual page token sets contain many
near-duplicate embeddings arising from background patches, whitespace, and repeated visual
motifs; early greedy selections therefore cover distinct query aspects without duplication,
and coverage loss per step is far below the worst-case allowance. These observations explain
why a conservative worst-case guarantee can coexist with high empirical recovery without
transferring the guarantee to nDCG.

\vspace{-0.02em}
\subsection{Reproducibility}
All benchmark tasks are public ViDoRe page-retrieval tasks. The paper specifies the retriever, pool factors, candidate-pool sizes, per-page budgets, BCR coverage setting, evaluation splits, latency units, and aggregate values used in every figure and table. Algorithms~\ref{alg:greedy} and~\ref{alg:topk}, together with the policy descriptions above, define the evaluated selectors; the timing scope is stated in Section~\ref{sec:budgets}.

\subsection{Limitations}
This study focuses on ColModernVBERT and ViDoRe-style page retrieval. The qualitative mechanism should apply to other multi-vector retrievers, including ColBERT-family retrievers and engines like PLAID, EMVB, and XTR \cite{santhanam2022plaid,nardini2024emvb,lee2023xtr}, but cross-backbone validation remains important. The two-tier design reduces the hot-path index, not total storage: it retains the full-token store, adds the PF32 index, and requires unbenchmarked shortlist I/O. The naive greedy implementation is computationally expensive; unmeasured directions include lazy greedy with marginal-gain upper bounds \cite{minoux1978accelerated}, candidate-pruned or vectorized greedy, and distillation into fast token policies. Our budget is fixed per page. Adaptive budgets could spend fewer tokens on easy candidates and more on visually dense or ambiguous pages.

\section{Conclusion}

Late-interaction visual document retrieval is effective because it preserves fine-grained page evidence, but storing and scoring that evidence is expensive. Our results show that this cost cannot be addressed by query-agnostic pooling alone. A compressed index can generate candidates efficiently, while query-aware allocation over the original tokens provides stronger evidence for final ranking. The MaxSim coverage formulation explains both why naive greedy recovers near-full quality as a reference policy and why token top-$k$ offers the practical low-latency choice. Across ten ViDoRe tasks, naive greedy recovers 98.39 percent of the PF1 score under a PF8-equivalent reranking budget, while token top-$k$ defines the deployable latency--quality frontier. Efficient late interaction is therefore not only a matter of storing fewer tokens, but of selecting the right tokens for each query.

\clearpage
\balance
\bibliographystyle{IEEEtran}
\bibliography{references}

\end{document}